\documentclass[reqno]{amsart}

\usepackage{graphicx}

\usepackage{amsthm}

\newcommand{\cal}[1]{\ensuremath{\mathcal{#1}}} 
\newcommand{\p}{\ensuremath{\varphi}}

\theoremstyle{plain}
\newtheorem{theorem}{Theorem}[section]
\newtheorem{prop}[theorem]{Proposition}
\newtheorem{lem}[theorem]{Lemma}
\newtheorem{cor}[theorem]{Corollary}

\theoremstyle{definition}
\newtheorem{defn}[theorem]{Definition}

\begin{document}

\title[Complete spectral data for analytic Anosov maps of the torus]
{Complete spectral data for analytic Anosov maps of the torus}

\author{J.~Slipantschuk, O.F.~Bandtlow, W.~Just}
\address{School of Mathematical Sciences, Queen Mary University of London,
Mile End Road, London E1 4NS, UK}
\email{J.Slipantschuk@qmul.ac.uk, O.Bandtlow@qmul.ac.uk, W.Just@qmul.ac.uk}

\date{9 May 2016}

\begin{abstract}
Using analytic properties of Blaschke factors we construct
a family of analytic hyperbolic diffeomorphisms of the torus for 
which the spectral
properties of the associated transfer operator acting on a 
suitable Hilbert space 
can be computed explicitly. As a result, we obtain explicit 
expressions for the decay of correlations
of analytic observables without resorting to any kind of
perturbation argument.
\end{abstract}


\maketitle

\section{Introduction}\label{sec1}

Spectral theory constitutes one of the major approaches to study 
complex chaotic motion. 
Drawing on both functional analytic techniques and dynamical systems
theory, it furnishes a powerful method to construct invariant measures 
with good statistical properties as well as 
a means to study the fine-structure of the corresponding 
correlation decay. The general theory is fairly 
well developed  (see, for example, \cite{KaHa:96,Kell:98,Bala:00}) 
and has resulted in several major breakthroughs in the understanding of complex
dynamical behaviour from an ergodic theoretic perspective.
Despite this deep understanding there is still a considerable lack of
exactly solvable models serving as paradigmatic examples 
illustrating the theory.  

To date, examples of maps for which spectral properties of the 
corresponding transfer operator can be 
computed explicitly are essentially limited to the one-dimensional uniformly 
expanding case, with the first examples arising in the context of 
piecewise linear Markov maps, where spectral properties can be reduced
to finite-dimensional
matrix calculations (see \cite{MoSoOs_PTP81}; see also \cite{SBJ1} 
for a more recent exposition). 
Exploiting the rich analytic structure of Blaschke products 
(see, for example, \cite{Mart_BLMS83})
nonlinear examples of full-branch analytic expanding interval maps for 
which complete spectral data of the corresponding transfer operator 
is available 
have recently been introduced by the authors (see
\cite{SlBaJu_NONL13}; 
see also \cite{BJS} for examples of analytic expanding maps of the circle). 

Trivial examples obtained by taking products of one-dimensional maps
excepted, the situation in higher dimensions is even more challenging, 
which is unfortunate, since 
diffeomorphisms, in particular higher-dimensional symplectic maps,  
play a vital role for the dynamical foundations of nonequilibrium
statistical mechanics, in particular regarding irreversibility
and entropy production \cite{AnTs_PA92,BaCo_JPA94,Dorf:98,Gall_CHA98}. 
Due to the lack of available models explicit calculations are normally
limited to the linear case, including the 
celebrated Arnold cat map or baker-type transformations. To the best
of our knowledge not a single properly 
nonlinear diffeomorphism is known for which
the entire spectrum and the corresponding correlation decay rates have 
been computed 
explicitly. We try to fill this gap by introducing a model where
all the spectral information is available.

For hyperbolic maps with expanding and contracting directions, 
progress was for a long time hampered by the lack of suitable 
function spaces on which the corresponding transfer operator can be shown to 
have good spectral properties. 
This changed with the publication of \cite{BKL}, where it was shown
that by adapting the space to take into account expanding and 
contracting directions the spectral properties of transfer operators 
familiar from the uniformly expanding situation can be retained for
Anosov diffeomorphisms of compact manifolds.  
Since then, quite a number of these `anisotropic' Banach spaces 
have been constructed (see \cite{GL1, GL2, BalT1, BalT2, BalG}),
capturing the behaviour of rather general hyperbolic diffeomorphisms with low
regularity. The main thrust of all these papers has been to show that
the associated transfer operator is quasicompact, that is, its peripheral
spectrum is discrete like that of a compact operator, but lower-lying
spectral points may (and usually will) be part of the essential
spectrum, characterised by persistence under compact perturbations. 

There are only few papers dealing with hyperbolic diffeomorphisms with
very high regularity, where there is a chance of obtaining compact
transfer operators. In the analytic setting, Rugh, in a paper
predating \cite{BKL}, has constructed anisotropic Banach spaces of
analytic functions on which the transfer operator of hyperbolic maps
with rather special geometries can be shown to be trace class, and
hence compact (see \cite{Rug}). 

In the following, we will introduce an example of an analytic
hyperbolic diffeomorphism of the torus, for which 
the 
entire spectrum of a properly defined 
compact transfer operator can be computed and 
linked with correlation decay of analytic 
observables. The underlying space is inspired by a recent study of 
Faure and Roy \cite{FaRo_NONL06}, who were able to link the
correlation decay of small analytic perturbations of linear 
automorphisms of the torus to spectral properties of a certain
transfer operator. While we still base our analysis on an 
analytic deformation of the cat map,  
we do not need to resort to a perturbative treatment, 
unlike \cite{FaRo_NONL06}. 

For any complex number $\lambda $ smaller than one in modulus let us introduce 
the analytic map $T:\mathbb{T}^2\rightarrow
\mathbb{T}^2$ on the complex unit torus $\mathbb{T}^2\subset
\mathbb{C}^2$ 
defined by
\begin{equation}\label{aa}
T(z_1,z_2)=\left(z_1 \frac{z_1-\lambda}{1-\bar{\lambda} z_1} z_2, 
\frac{z_1-\lambda}{1-\bar{\lambda} z_1} z_2 \right) \, .
\end{equation}
Using canonical coordinates on the unit torus, $z_\ell=\exp(i\phi_\ell)$, the 
real representation of the map reads
\begin{equation}\label{ab}
(\phi_1, \phi_2)\mapsto (2 \phi_1 + \psi(\phi_1) + \phi_2, \phi_1 + \psi(\phi_1)+\phi_2)\,,
\end{equation}
where the nonlinear part is given by
\begin{equation}\label{ac}
\psi(\phi)=2 \arctan \left(\frac{|\lambda|\sin(\phi-\alpha)}{1-|\lambda|\cos(\phi-\alpha)}\right) \,.
\end{equation}
Here, $\lambda=|\lambda|\exp(i\alpha)$ denotes the polar 
representation of the 
parameter with $|\lambda|<1$. Clearly, our family of maps 
contains the Arnold 
cat map for the choice $\lambda=0$. The toral map (\ref{aa}) is a special case of a so-called
two-dimensional Blaschke product which has already received some attention in the context of 
ergodic theory (see \cite{PuSh_ETDS08}). 

It is not difficult to see that the derivative of the map 
given by (\ref{ab}) maps the first and third quadrant of 
$\mathbb{R}^2$ strictly inside itself and that the derivative of its
inverse maps the second and fourth quadrant of $\mathbb{R}^2$ strictly
inside itself. Thus (\ref{aa}) yields a family of 
analytic uniformly hyperbolic toral diffeomorphisms,
also known as Anosov diffeomorphisms (see \cite[Lemma~4]{Mose_JDE69}
or \cite[Chapter~2.1.b]{hasselblatt}).

Clearly, the map defined
by (\ref{ab}) is area-preserving and thus provides an example of a  
chaotic Hamiltonian system. A few more
empirical features, numerical simulations, and some basic 
results on correlation decay are presented in Appendix~\ref{appa}.

Unlike the situation for one-dimensional 
non-invertible maps there is no clear distinction between
Perron-Frobenius and Koopman operators as we are dealing with area preserving
diffeomorphisms. The operator governing the dynamics of our system is 
essentially 
a composition operator $\mathcal C$ defined by 
\begin{equation}\label{ad}
(\mathcal{C} f)(z_1,z_2) = (f \circ T)(z_1,z_2)=
f\left(z_1 \frac{z_1-\lambda}{1-\bar{\lambda} z_1} z_2, 
\frac{z_1-\lambda}{1-\bar{\lambda} z_1} z_2 \right)
\end{equation}
where $f: \mathbb{T}^2 \rightarrow \mathbb{C}$. 
As alluded to earlier, the choice of a space of functions on which
$\mathcal C$ acts and has nice spectral properties is a delicate
matter. 
We shall use a family of 
Hilbert spaces $\mathcal{H}_a$ indexed by 
a positive real parameter $a$ 
which contains all 
Laurent polynomials on the unit 
torus as a dense subset. 
Postponing the formal definition to the following section our main result can 
be stated as follows. 
\begin{theorem}\label{propa}
The composition operator $\mathcal{C}: \mathcal{H}_a \rightarrow \mathcal{H}_a$ 
is a well-defined compact operator for any $a>0$ and any $|\lambda|<1$. Its 
spectrum is given by
\begin{equation}\label{ae}
\sigma(\mathcal{C}) = \{ (-\lambda)^n: n\in\mathbb{N} \} \cup 
\{ (-\bar{\lambda})^n: n\in\mathbb{N} \}
\cup \{1,0\} \,.
\end{equation}
Each non-zero element of the spectrum is an eigenvalue the algebraic multiplicity of which 
coincides with the number of times the non-zero number occurs in 
(\ref{ae}). 
\end{theorem}
Using the definition (\ref{ad}) and the invariance of Haar measure
$\mu$ on the unit torus
it is straightforward to relate the spectral properties of the operator with
correlation functions and to bound the decay of correlations for sufficiently
nice observables. 
\begin{cor}\label{cora}
For any functions $g: \mathbb{T}^2 \rightarrow \mathbb{C}$ and $h: \mathbb{T}^2 \rightarrow \mathbb{C}$
analytic in an open neighbourhood of the unit torus 
the corresponding correlation function
\begin{equation}\label{af}
C_{gh}(k) = \int g \cdot  h\circ T^k d \mu - \int g d\mu \int h d\mu\,,
\end{equation}
where $\mu$ denotes the invariant Haar measure on $\mathbb{T}^2$,
satisfies 
\begin{equation}
\limsup_{k\to \infty} |C_{gh}(k)|^{1/k} \leq |\lambda|\,.
\end{equation}
In particular, $T$ is strongly mixing with respect to $\mu$. 
\end{cor}
With a little bit more effort one can also derive asymptotic expansions
for the correlation function. In particular, the estimate given in 
Corollary~\ref{cora} is sharp as one can easily find cases 
where the upper bound is attained, see (\ref{fd}).

Another simple consequence of Theorem~\ref{propa} is the following 
result on the location of the Pollicott-Ruelle resonances 
(see \cite{Pol1,Pol2,Rue1,Rue2}) of $T$, that is, 
the poles of the meromorphic continuation of the Z-transform of the correlation
function. 
\begin{cor}
\label{corb}
For any 
$g: \mathbb{T}^2 \rightarrow \mathbb{C}$ and $h: \mathbb{T}^2 \rightarrow \mathbb{C}$
analytic in an open neighbourhood of the unit torus 
the Z-transform $\hat{C}_{gh}$ 
of the corresponding correlation function given by 
\begin{equation}
\hat{C}_{gh}(\zeta) = \sum_{k=0}^\infty \zeta^{-k} C_{gh}(k)
\end{equation}
for $\zeta\in\mathbb{C}$ with $|\zeta|>1$, 
has a meromorphic continuation to $\mathbb{C}\setminus \{0\}$ with
poles at 
\begin{equation}
\{ (-\lambda)^n: n\in\mathbb{N} \} \cup 
\{ (-\bar{\lambda})^n: n\in\mathbb{N} \}\,.
\end{equation}
\end{cor}

This article is organised as follows. 
In Section~\ref{sec2} we will define the function 
space $\mathcal{H}_a$ on which the 
transfer operator (\ref{ad}) will be defined. We will spend some effort 
on its motivation, as 
its structure is fundamentally
linked to the physics of the underlying dynamical system. 
Compactness of the 
composition operator will be proven in 
Section~\ref{sec3} by 
establishing suitable 
bounds on the entries of a matrix 
representation of $\mathcal{C}$ with respect to an orthonormal basis 
of $\mathcal{H}_a$. Using the fact that this matrix representation is 
upper-triangular 
we will then be able to 
obtain the entire spectrum of $\mathcal C$ in closed form, 
thus completing the proof of our main result, Theorem~\ref{propa}. 

Section~\ref{sec4} is devoted to proving the two corollaries, which
involves a discussion of the properties of the 
invariant measure and the corresponding correlation decay for analytic 
observables. 

Part of our presentation requires some basic knowledge of 
functional analysis, which, in spite of the fact that it 
can be found in standard textbooks, we cover in some detail 
so as to make the 
exposition accessible to a larger audience 
in applied dynamical systems theory. 

In this article, we shall only be concerned with the particular
example 
given in (\ref{aa}), postponing 
the discussion of possible generalisations to the conclusion and 
Appendix~\ref{appc}.

\section{Hilbert space and transfer operator}\label{sec2}
The main purpose of this section is to introduce a family of Hilbert
spaces and to show that the composition operator (\ref{ad}) is compact
on each of these spaces. 

We start by defining the family of Hilbert spaces. 
For $\lambda=0$ the map given by (\ref{aa}) or (\ref{ab}) reduces to a 
linear automorphism 
on the torus.
The corresponding unstable/stable eigenvalues and eigenvectors are given by
\begin{equation}\label{ba}
\lambda_{u/s}=\p^{\pm 2}, \quad v_{u/s}=(\lambda_{u/s}-1,1) \,,
\end{equation}
where $\p:=(1+\sqrt{5})/2$ denotes the golden mean. 
For brevity we will use multi-index notation $n=(n_1,n_2)\in \mathbb{Z}^2$ with 
$|n|=|n_1|+|n_2|$, and we abbreviate the monomials of $z=(z_1,z_2)\in
\mathbb{C}^2$ by $z^n=z_1^{n_1} z_2^{n_2}$. Let us denote by
\begin{equation}\label{bb}
n_{u/s}=v_{u/s} n = (\lambda_{u/s}-1)n_1+n_2
\end{equation}
the components of $n$ with respect to the stable and unstable direction of 
the cat map. 

Before defining the family of spaces recall that a 
\emph{Laurent polynomial} is a finite linear combination of monomials
$z\mapsto z^n$ where $n\in \mathbb{Z}^2$. The set of all Laurent
polynomials will be denoted by $\mathcal{L}$. Thus 
\begin{equation}
\label{calL}
\mathcal{L}=\{f:\mathbb{T}^2\to \mathbb{C}: 
f(z)=\sum_{|n|\leq N}f_nz^n, 
\text{ with $f_n\in\mathbb{C}, N\in \mathbb{N}$} \}\,.
\end{equation}
Motivated by \cite{FaRo_NONL06}, we will define
anisotropic Hilbert spaces adapted to the action of the transfer
operator given by (\ref{ad}) as the completion of the set of Laurent polynomials
with respect to a certain norm, which we shall define presently. 

Given $a>0$ define an inner product on $\mathcal{L}$ by 
\begin{equation}
\label{Haip}
\langle f, g  \rangle_a = 
\sum_{n \in \mathbb{Z}^2} f_n \bar{g}_n \exp(-2a|n_u|+2 a |n_s|) \,,
\end{equation}
where $(f_n)_{n\in \mathbb{Z}^2}$ and 
$(g_n)_{n\in \mathbb{Z}^2}$ denote the Fourier coefficients of the
Laurent polynomials $f$ and $g$, respectively, that is, 
\begin{equation}
f(z)=\sum_{n\in \mathbb{Z}^2} f_nz^n \text{ and }
g(z)=\sum_{n\in \mathbb{Z}^2} g_nz^n\,;
\end{equation}
the corresponding norm will be denoted by 
$\|\cdot\|_a$, that is, we have 
\begin{equation}\label{bd}
\| f \|_a^2 = \sum_{n \in \mathbb{Z}^2} |f_n|^2 \exp(-2a|n_u|+2 a
|n_s|)\,. 
\end{equation}
We are now ready to define the family of Hilbert spaces. 
\begin{defn}\label{defa}
Let $a>0$ then $\mathcal{H}_a$ is the completion of $\mathcal{L}$ with
respect to the norm $\|\cdot \|_a$. 
\end{defn}
It turns out that $\mathcal{H}_a$ is a separable 
Hilbert space, which, by construction, contains all Laurent
polynomials as a dense subset (see, for example,
\cite[Theorem~I.3]{RS}). However, it also contains functions analytic
in a sufficiently large open neighbourhood of the torus as the following
lemma shows. 
\begin{lem}
\label{analem}
If $a>0$ and $f:\mathbb{T}^2\to \mathbb{C}$ is analytic in an open 
neighbourhood of 
\begin{equation}
\label{annulus}
\{ \exp(-\sqrt{5}a)\leq |z_1| \leq \exp(\sqrt{5}a) \} \times 
\{ |z_2| =1 \}\,,
\end{equation}
then $f\in \mathcal{H}_a$. In particular, any function analytic on an
open 
neighbourhood of the torus belongs to $\mathcal{H}_a$ for all
sufficiently small $a$.  
\end{lem}
\begin{proof}
Suppose that $f$ is analytic on an open neighbourhood of 
the poly-annulus (\ref{annulus}). Then, using
Cauchy's integral formula, the function $f$ has
a Laurent expansion of the form 
\begin{equation} 
f(z)=\sum_{n\in \mathbb{Z}^2}f_nz^n
\end{equation} 
with  
\begin{equation}
\label{bounded}
\sum_{n\in \mathbb{Z}^2}|f_n|^2 \exp(2a\sqrt{5}|n_1|)<\infty\,.
\end{equation}
Since, by (\ref{bb}), we have 
\begin{equation}
|n_s|-|n_u|\leq |n_s-n_u|=\sqrt{5}|n_1|\,,
\end{equation}
the bound (\ref{bounded}) implies that 
\begin{multline}
\sum_{n\in \mathbb{Z}^2}|f_n|^2 \exp(-2a|n_u|+2a|n_s|)\\
=\sum_{n\in \mathbb{Z}^2}|f_n|^2\exp(2a\sqrt{5}|n_1|) 
\exp(-2a|n_u|+2a|n_s|-2a\sqrt{5}|n_1|)<\infty\,,
\end{multline}
which shows that $f$ is a limit of Laurent polynomials convergent in the 
norm $\|\cdot\|_a$ and can thus be uniquely identified with an
element in $\mathcal{H}_a$. 
\end{proof}
While, as we have just seen, the space $\mathcal{H}_a$ contains
functions analytic on a sufficiently large neighbourhood of the torus,
it also contains generalised functions, 
not interpretable as ordinary functions on the torus.

At first glance, the choice of weighting in the definition of the norm
(\ref{bd}) appears peculiar. 
However, this choice is intimately linked with the underlying 
dynamics. Broadly speaking, 
the weighting requires that the Fourier
coefficients $(f_n)_{n\in \mathbb{Z}^2}$ 
of $f\in \mathcal{H}_a$ decay exponentially in the 
stable direction whereas they are 
allowed to grow exponentially in the unstable direction. 
The corresponding function
on the unit torus inherits this behaviour, that is, it is smooth in the 
stable but allowed to be rather rough in the unstable direction. 
It is precisely this property 
which makes it possible 
to capture the dynamics of the underlying map. For instance, the
simple 
textbook 
example of the
Arnold cat map shows that an initially smooth density remains smooth 
along the unstable direction but
becomes jagged in the stable direction. If we keep in mind that a 
Perron-Frobenius operator 
governing the motion of densities involves the inverse of the map and thus 
interchanges stable 
and unstable direction it is precisely the space defined above which 
is able to capture the 
ergodic properties of the dynamical system. In physics terms, the 
structure of this space 
breaks the time reversal symmetry of the dynamical system, capturing
the macroscopically  
irreversible behaviour of the motion \cite{AnTs_PA92}. 

As in \cite{FaRo_NONL06}, 
we could have used a slightly more general setup for the underlying space  
by giving different weights to the stable and unstable parts in 
Definition~\ref{defa}.  
The restricted 
case considered here
will turn out to be sufficient for our purpose. 
We will revisit this issue in the conclusion.

For later use, we note that the normalised monomials 
\begin{equation}\label{bg}
e_n(z)=z^n  \exp(a|n_u|- a |n_s|) \quad (\forall n\in \mathbb{Z}^2) \,,
\end{equation}
yield an orthonormal basis
for $\mathcal{H}_a$ for every $a>0$. 

Having introduced the underlying Hilbert space we are now going to 
define a 
transfer operator associated with the map. 
The definition (\ref{ad}) makes perfect sense for 
Laurent polynomials, which form a dense subset of
$\mathcal{H}_a$. Hence, 
it remains to show that
$\mathcal{C}$ is bounded with respect to the norm of $\mathcal{H}_a$ 
on the set of Laurent polynomials.  For this in turn,  
it is sufficient
to evaluate the images (under $\mathcal{C}$) of the basis elements
(\ref{bg}) 
and to 
show that their norm decays
sufficiently fast with $|n|$. 

We start by observing that (\ref{ad}) and (\ref{bg}) yield
\begin{align}\label{bh}
(\mathcal{C}e_n)(z) =&  \exp(a|n_u|- a |n_s|) z_1^{n_1} z_2^{n_1+n_2} \left(
\frac{z_1-\lambda}{1-\bar{\lambda} z_1}\right)^{n_1+n_2}\nonumber \\
=& \exp(a|n_u|- a |n_s|)  \sum_{m\in \mathbb{Z}^2} b_{n,m} z^m\,,
\end{align}
where the expansion coefficients of the Laurent series in a neighbourhood of the
unit torus are given by 
\begin{equation}\label{bi}
b_{n,m}=\delta_{n_1+n_2,m_2} M_{-m_1+2n_1+n_2}(\lambda,n_1+n_2) \,, 
\end{equation}
where we have introduced the shorthand 
\begin{equation}\label{bj}
M_\ell(\lambda,k)= \int_{|\zeta|=1}
\zeta^\ell \left(\frac{1-\lambda/\zeta}{1-\bar{\lambda} \zeta}\right)^k 
\frac{d\zeta}{2 \pi i \zeta}
\end{equation}
for the expansion coefficient of a single Blaschke factor.
Hence, using the definition of the norm in (\ref{bd}), we obtain
\begin{align}\label{bk}
\| \mathcal{C} e_n \|_a^2 =& \sum_{m\in \mathbb{Z}^2} 
\exp(2 a|n_u|- 2 a |n_s|) \delta_{n_1+n_2,m_2}
|M_{-m_1+2 n_1+n_2}(\lambda,n_1+n_2)|^2 \nonumber \\
 & \qquad \times \exp(-2a|m_u|+2 a |m_s|) \, .
\end{align}
Before we proceed, let us first comment on the trivial case 
of the cat map, which corresponds to $\lambda=0$. 
In this case, expression (\ref{bj}) simplifies to $M_{\ell}(\lambda, k)=\delta_{\ell, 0}$ and only 
the term
$m_1=2n_1+n_2$, $m_2=n_1+n_2$ contributes to the series in (\ref{bk}). 
Thanks to (\ref{bb}), that is, thanks to the stable and unstable directions of the cat map
this gives $m_{u/s}=\lambda_{u/s} n_{n/s}$ and so (\ref{bk}) becomes 
\begin{equation}
\label{lambdanull}
\|\mathcal{C}e_n\|_a^2=\exp(-2a(\lambda_u-1)|n_u|-2a(1-\lambda_s) |n_s|)\,.
\end{equation}
Since all norms in $\mathbb{R}^2$ are equivalent, we see that 
in this simple case there is a $\delta>0$, such that 
\begin{equation}
\label{lambdanulldelta}
\|\mathcal{C}e_n\|_a\leq \exp(-\delta |n|)\quad (\forall n \in \mathbb{Z}^2)\,,
\end{equation}
that is, we end up with an upper bound, which is exponentially small in $|n|$. This in turn finally guarantees 
that the transfer operator $\mathcal{C}$ is well-defined and compact on $\mathcal{H}_a$, using a simple summability argument (see the proof of Proposition~\ref{prop:comp}). 

The same observation together with a localisation argument for the expression
(\ref{bj}) has been used in \cite{FaRo_NONL06} to derive 
similar upper bounds for the transfer operators of maps which are small perturbations (in the $C^1$ sense) 
of linear maps of the torus. Restricting to our particular choice of maps, 
we will obtain a slightly stronger result without resorting to any perturbative argument. 

In order to do this, let us first focus on an estimate for the expression (\ref{bj}). 
\begin{lem}\label{lema}
For any $\lambda=|\lambda| \exp(i\gamma)\in \mathbb{C}$ with $|\lambda|<1$ 
the expression (\ref{bj}) obeys
\begin{itemize}
\item[i)] $M_\ell(\lambda,k)=M_{-\ell}(\bar{\lambda},-k)$;
\item[ii)] $M_\ell(|\lambda|\exp(i\gamma),k)=\exp(i\ell\gamma) M_\ell(|\lambda|,k)$;
\item[iii)] $M_\ell(\lambda,0)=\delta_{\ell,0}$;
\item[iv)] $M_\ell(\lambda,k)=0$ if $\ell>k>0$;
\item[v)] $|M_\ell(\lambda,k)|\leq 1$.
\end{itemize}
In addition, there exists $\alpha>0$ and $\beta\in (0,1)$ such that for $k>0$ and
$\beta k \leq \ell \leq k$ the estimate
\begin{equation}\label{bl}
|M_\ell(\lambda,k)|\leq \exp(-\alpha(\ell-\beta k))
\end{equation}
holds.
\end{lem}
\begin{proof}
The symmetry properties i) and ii) can be obtained by appropriate 
substitutions
in the integral (\ref{bj}), namely $\zeta'=\zeta^{-1}$ and 
$\zeta'=\zeta \exp(-i\gamma)$,
respectively. Property iii) is obvious. Since the integrand in (\ref{bj}) 
is holomorphic
in the unit disk for  $\ell>k>0$, property iv) follows. Finally v) is 
obvious, as the integrand
is bounded by one. Hence, the only non-trivial part which remains to be 
proven is the  estimate (\ref{bl}). 

Due to the phase symmetry ii) it is sufficient to prove (\ref{bl}) 
with $|\lambda|$ instead
of $\lambda$. 
By contour deformation we have for $r\in (0,1)$
\begin{align}\label{bm}
|M_\ell(\lambda,k)| =& \left| \int_{|\zeta|=r}
\zeta^\ell \left(\frac{1-|\lambda|/\zeta}{1-|\lambda| \zeta}\right)^k 
\frac{d\zeta}{2 \pi i   \zeta} \right| \nonumber \\
\leq& \frac{r^\ell}{2 \pi} \int_0^{2 \pi} \left|
\frac{1-|\lambda|/r \exp(-i\phi)}{1-|\lambda| r \exp(i\phi)} \right|^k
d \phi \, .
\end{align}
It is not difficult to see that the integrand takes 
its maximum at $\phi=\pi$, that is, 
\begin{equation}\label{bn}
\left|
\frac{1-|\lambda|/r \exp(-i\phi)}{1-|\lambda| r \exp(i\phi)} \right|
\leq \frac{1+|\lambda|/r}{1+|\lambda| r}
\quad (\forall \phi \in [0,2\pi))
\end{equation}
so that for any $\beta\in (0,1)$ we have
\begin{equation}\label{bo}
|M_\ell(\lambda,k)| \leq r^{\ell-\beta k} 
\left(r^\beta \frac{1+|\lambda|/r}{1+|\lambda| r}\right)^k \, .
\end{equation}
The base $F(r)=r^\beta (1+|\lambda|/r)/(1+|\lambda| r)$ clearly obeys
$F(1)=1$ and $F'(1)>0$ if $\beta>2 |\lambda|/(1+|\lambda|)$. Hence,
the assertion follows by first choosing $\beta \in (2 |\lambda|/(1+|\lambda|), 1)$ 
and then choosing $r=\exp(-\alpha)\in(0,1)$ with $F(r)\leq 1$. 
\end{proof}
Let us now return to (\ref{bk}). Using Lemma~\ref{lema}
it is fairly straightforward to establish the following. 
\begin{lem}\label{lemb}
For $\lambda\in\mathbb{C}$ with $|\lambda|<1$ and $a>0$
there exists $c>0$ and $\delta>0$ such that
\begin{equation}\label{bp}
\| {\cal C} e_n \|_a \leq c \exp(-\delta |n|)
\quad (\forall n\in \mathbb{Z}^2)\,.
\end{equation}
\end{lem}
\begin{proof}
Because of the symmetry relations i) and ii) in Lemma \ref{lema} 
the series \eqref{bk} obeys $\|{\cal C}e_n\|_a=\|{\cal C} e_{-n}\|_a$.
Hence it is sufficient to consider the case $n_1+n_2\geq 0$.

If $n_1+n_2=0$ then property iii) of Lemma \ref{lema} guarantees that
only a single term with $m_1=2n_1+n_2$ and $m_2=n_1+n_2$ 
contributes to \eqref{bk}, so that 
$m_u=\lambda_u n_u$ and $n_s=\lambda_s n_s$ by \eqref{bb}. Thus
\begin{align}\label{bq}
\|{\cal C}e_n\|_a =& \exp(-a(\lambda_u-1) |n_u|-a(1-\lambda_s)|n_s|)
\nonumber \\
\leq& \exp(-a(1-\lambda_s)(|n_u|+|n_s|)),
\end{align}
where we have used $\lambda_u\lambda_s=1$ and $\lambda_u>1$. As $|n_1|+|n_2|\leq
2(|n_u|+|n_s|)$ relation (\ref{bp}) holds for any $c\geq 1$ and 
any $\delta\leq a(1-\lambda_s)/2$.

Let us now assume that $n_1+n_2>0$. The sum in \eqref{bk} only runs over $m_1$, 
as only $m_2=n_1+n_2$ can give rise to non-zero terms.
Making use of (\ref{bl}), we now split this sum into three parts. 
\begin{equation}
\|\mathcal{C}e_n\|^2_a = S_1 + S_2 + S_3,
\end{equation}
where
\begin{equation}\label{eq:S_i}
S_i = \sum_{\substack{m_1\in I_i\\m_2=n_1+n_2}} |M_{-m_1+2n_1+n_2}(\lambda,n_1+n_2)|^2 
\exp(-2a(|n_s|-|n_u|+|m_u|-|m_s|))
\end{equation}
with $I_1 = \{m_1: m_1 < n_1\}$, $I_2 = \{m_1: n_1 \le m_1 \le n_1 + (1-\beta)(n_1+n_2)\}$, 
and $I_3 = \{m_1:m_1 >  n_1 + (1-\beta)(n_1+n_2)\}$.
Note that $S_1 = 0$ by iv) of Lemma~\ref{lema}.

For $S_2$ and $S_3$, we first 
need to have a closer look at the exponential factor.
Using \eqref{bb}, the exponent can be written as
\begin{equation}\label{bs}
2a(|n_s|-|n_u|+|m_u|-|m_s|) = a |n_1+n_2| F\left(\frac{m_1}{n_1+n_2}, 
\frac{n_1}{n_1+n_2}\right)
\end{equation}
where
\begin{equation*}
F(x,y)=2\left (
|\p x+1|-|\p^{-1}x-1| +|\p y-1|-
|\p^{-1}y+1| \right )\,,
\end{equation*}
and, as before, $\p=(\sqrt{5}+1)/2$ denotes the golden mean. 
If we employ 
the basic lower bound for $F$ derived in Lemma~\ref{lemd} in
Appendix~\ref{appb}, then \eqref{bs} yields
\begin{equation}\label{bu}
2a(|n_s|-|n_u|+|m_u|-|m_s|) \geq a(m_1-n_1+ |n_1| \Theta_n/2)
\end{equation}
with 
\begin{equation}\label{bv}
\Theta_n=
 \begin{cases}
   0 & \text{if $|n_1| < 2\p^{-1}|n_1+n_2|$,}\\
   1 & \text{if $|n_1| \geq 2\p^{-1}|n_1+n_2|$.}
\end{cases}
\end{equation}

With this lower bound we can now estimate $S_2$ and $S_3$ as we
have essentially reduced the problem to a geometric series.
For $S_3$ we use the trivial estimate v) of Lemma \ref{lema} giving  
\begin{equation}\label{bw}
|S_3| \leq \frac{\exp(-a(1-\beta)(n_1+n_2))}{1-\exp(-a)} 
\exp(-a |n_1| \Theta_n/2) \,.
\end{equation}
Now, a short calculation shows that 
\begin{equation}\label{eq:n1n2}
|n_1+n_2|+|n_1|\Theta_n \geq (|n_1|+|n_2|)/4\,,
\end{equation}
so we can bound $S_3$ from above by 
\begin{equation}\label{S3above}
|S_3|\leq \frac{\exp(-\delta'|n|)}{1-\exp(-a)}
\end{equation}
for any $\delta'\leq a(1-\beta)/8$. 

For $S_2$, the bound \eqref{bl} yields
\begin{equation*}
|S_2| \leq \sum_{0\leq k \leq (1-\beta)(n_1+n_2)}
\exp(-ak) \exp(-2\alpha((1-\beta)(n_1+n_2)-k)\exp(-a |n_1|\Theta_n/2).
\end{equation*}
Estimating this finite sum by a simple bound for its largest term 
we can write
\begin{equation*}
|S_2| \leq ((1-\beta)(n_1+n_2)+1) \exp(-\min\{a,2\alpha \}(1-\beta)(n_1+n_2)
-a |n_1|\Theta_n/2).
\end{equation*}
Using  (\ref{eq:n1n2}) we see that 
\begin{equation} \label{S2above}
|S_2| \leq ((1-\beta)|n|+1) \exp(-\delta'|n|)\,,
\end{equation}
for any  $\delta'\leq \min\{ a/8,\alpha/4 \} (1-\beta)$. 
Putting the two bounds for $S_2$ and $S_3$ together, the assertion finally follows.  
\end{proof}

Standard arguments now yield the following result.

\begin{prop}\label{prop:comp}
For any $\lambda\in\mathbb{C}$ with $|\lambda|<1$ and any $a>0$, expression \eqref{ad} gives rise 
to a bounded and compact operator $\mathcal{C}\colon \mathcal{H}_a \to \mathcal{H}_a$. 
\end{prop}

\begin{proof}
Lemma \ref{lemb} implies that 
\begin{equation}
\label{CHS}
M:=\left ( \sum_{n\in \mathbb{Z}} \|\mathcal{C} e_n\|_a^2 \right )^{1/2} <\infty\,.  
\end{equation}
Thus the operator given by the expression \eqref{ad} is bounded on the
set of Laurent polynomials since, using the Cauchy-Schwarz 
inequality, we have    
\begin{equation}
\|\mathcal{C} f\|_a \leq \sum_{n\in\mathbb{Z}} |f_n| \exp(-a|n_u|+ a |n_s|)  \|\mathcal{C} e_n\|_a
\leq M\| f\|_a\,;
\end{equation}
thus, by a standard result (see, for example, \cite[Theorem~1.7]{RS}) 
the operator $\mathcal{C}$ has a unique bounded extension, which
we denote by the same symbol, from $\mathcal{H}_a$ to
$\mathcal{H}_a$. In fact, inequality (\ref{CHS}) implies 
that $\mathcal{C}$ is Hilbert-Schmidt on $\mathcal{H}_a$, and
therefore 
compact (see, for example, \cite[Theorem~VI.22]{RS}).
\end{proof}
We have just seen that $\mathcal{C}:\mathcal{H}_a\to \mathcal{H}_a$ 
is Hilbert-Schmidt. 
In fact, the exponential decay of the matrix elements of 
$\mathcal{C}$ established in Lemma~\ref{lemb} implies that 
$\mathcal{C}$ has even stronger compactness properties. 
It can be shown, for example by the argument used in the proof of 
\cite[Theorem~7]{FaRo_NONL06}, 
that the singular values of $\mathcal{C}$ decay
at a stretched-exponential rate, so $\mathcal{C}$ belongs the
exponential classes introduced in \cite{expo}, in common with the
transfer operators corresponding to higher-dimensional analytic
expanding maps (see, for example, \cite{BJ1,BJ2}).

\section{Spectral data}\label{sec3}

In order to complete the proof of Theorem~\ref{propa}, it remains to compute the spectrum
of the operator $\mathcal{C}$. This can be achieved by considering suitable matrix
representations of projections of this compact operator to finite-dimensional subspaces.

We start by observing that, by \eqref{bg} and \eqref{bh}, the 
matrix representation $\Gamma$ of 
$\mathcal{C}$ with respect to the orthonormal basis $(e_n)_{n\in \mathbb{Z}^2}$ of $\mathcal{H}_a$ is 
of the form
\begin{equation}\label{ca}
\Gamma_{n,m}=\langle \mathcal{C} e_n\,e_m \rangle_a =
b_{n,m}  \exp(a|n_u|-a|n_s|-a |m_u|+a|m_s|) \,.
\end{equation}

A short calculation using \eqref{bi}, \eqref{bj} and Lemma~\ref{lema} i), iii), iv)
yields the following cases:
\begin{eqnarray}
n_1+n_2\neq m_2: &\quad & b_{n,m}=0 \label{cb} \\
n_1+n_2= m_2=0: &\quad & b_{n,m}=\delta_{m_1, n_1} \label{cc} \\
n_1+n_2= m_2>0: &\quad & b_{n,m}=\left\{
\begin{array}{lcl} 0 & \mbox{ if } & m_1<n_1\\
(-\lambda)^{m_2} &\mbox{ if } & m_1=n_1
\end{array} \right. \label{cd}\\
n_1+n_2= m_2<0: &\quad & b_{n,m}=\left\{
\begin{array}{lcl} 0 & \mbox{ if } & m_1>n_1\\
(-\bar{\lambda})^{-m_2} &\mbox{ if } & m_1=n_1 \, .
\end{array} \right. \label{ce}
\end{eqnarray}

These properties will turn out to be sufficient to show that we can
order the basis elements in such a way that the corresponding matrix
is upper-triangular. We first arrange the basis 
$(e_n)_{n\in \mathbb{Z}^2}$ 
as a sequence in the order of increasing norm
$|n|$, with groups of elements with the same norm traversed in 
counter-clockwise direction, that is, 
\begin{equation*}
e_{0,0}, e_{1,0}, e_{0,1}, e_{-1,0}, e_{0,-1}, e_{2,0}, 
e_{1,1}, e_{0,2}, e_{-1,1}, e_{-2,0}, e_{-1,-1}, e_{0,-2}, e_{1,-1}, \ldots \, .
\end{equation*}
We then re-order this sequence as follows. We move along the sequence above from left to right. If we encounter a basis element $e_{n_1,n_2}$ with $n_1 n_2 < 0$ we move the element to the left-most position of the current sequence. We thus obtain the following order: 
\begin{equation*}
\ldots, e_{1,-1}, e_{-1,1}, e_{0,0}, e_{1,0}, e_{0,1}, e_{-1,0}, e_{0,-1}, e_{2,0}, 
e_{1,1}, e_{0,2}, e_{-2,0}, e_{-1,-1}, e_{0,-2}, \ldots \, .
\end{equation*}
\begin{lem}\label{lemc}
The matrix given by \eqref{ca} 
is upper-triangular with respect to the basis
re-ordered as above. Moreover, its only non-zero diagonal entries are 
$\Gamma_{00,00}=1$, $\Gamma_{0k,0k}=(-\lambda)^k$ and 
$\Gamma_{0 -k,0 -k}=(-\bar{\lambda})^k$ where $k\in \mathbb{N}$. 
\end{lem}
\begin{proof}
We first prove that the entire lower-left block with $n_1 n_2 \geq 0$ and 
$m_1 m_2 < 0$ consists of zeros. 
Assume the contrary, that is, assume that there exists 
some non-vanishing matrix 
element $\Gamma_{n,m}$ in this sector, that is, 
$b_{n,m}\neq 0$. From \eqref{cb} we get 
$m_2 = n_1 + n_2$, which is non-zero as $m_1m_2\neq 0$. 
If $n_1 + n_2 = m_2 > 0 > m_1$, then $n_1, n_2$ are non-negative, so $n_1>m_1$
and \eqref{cd} results in the contradiction $b_{n,m} = 0$. 
A similar reasoning applies in the case $n_1 + n_2 = m_2 < 0 < m_1$.

Next, we confirm that the upper-left block matrix with $n_1 n_2 < 0$ and 
$m_1 m_2 < 0$ is an upper-triangular matrix with zeros on the diagonal. For 
this we assume that a matrix entry lying on or below the diagonal 
is non-zero. Note that, with the chosen ordering, the indices of a
matrix 
element 
on or below the diagonal satisfy 
\begin{equation}\label{ch}
|m_1| + |m_2| = |m| \geq |n| = |n_1| + |n_2| \, .
\end{equation}
Since $m_1 m_2 < 0$ we have $m_2 \neq 0$. If $m_2 > 0$ then \eqref{cd} 
implies $m_1 \geq n_1$ and
hence the condition (\ref{ch}) results in 
$n_1 \leq m_1 < 0 < n_2 \leq m_2$. In particular
$m_2 > n_1 + n_2$ so that \eqref{cb} yields the contradiction $b_{n,m}=0$. 
The case for $m_2 < 0$ is analogous.

Finally, we show the claim for the most interesting case, 
the lower-right block matrix where 
$n_1 n_2 \geq 0$ and 
$m_1m_2 \geq 0$. In this case, the indices of a matrix element  
on or below the diagonal satisfy $|n| \geq |m|$. Since 
the components of $n$ and $m$ have
equal signs, this condition can be written as $|n_1 + n_2| \geq  
|m_1 + m_2|$. If 
$b_{n,m}\neq 0$, then $m_2 = n_1 + n_2$ by \eqref{cb},
so that $|m_2| \geq |m_1 + m_2|$. Since $m_1 m_2 \geq 0$ we
conclude $m_1 = 0$. Then one of the following three cases holds:

i) $m_2 = 0$: We get $0 = m_2 = n_1 + n_2$ and $n_1 n_2 \geq 0$ results in 
$n_1 = n_2 = m_1 = m_2 = 0$ for which
$\Gamma_{00,00} = 1$ by \eqref{cc}. 

ii) $m_2 > 0$: By \eqref{cd} we have $0 = m_1 \geq n_1$.
Since $m_2 = n_1 + n_2>0$ and $n_1$ and $n_2$ have the same sign this implies 
$m_1 = n_1 = 0$
and $m_2 = n_2$. The corresponding diagonal entry is given by  
$\Gamma_{0 m_2,0 m_2} = (-\lambda)^{m_2}$ by \eqref{cd}. 

iii) $m_2 < 0$: By \eqref{ce} we have $0 = m_1 \leq n_1$ and by the
same argument as in the previous case we get $n_1 = m_1 = 0$ and 
$n_2 = m_2$. The
corresponding diagonal entry is given by $\Gamma_{0 m_2,0 m_2} = 
(-\bar{\lambda})^{-m_2}$ by \eqref{ce}. 
\end{proof}

We are now able to finish the proof of the main result.

\begin{proof}[Proof of Theorem~\ref{propa}]
Compactness of $\mathcal{C}\colon \mathcal{H}_a \to \mathcal{H}_a$ was 
established in
Proposition \ref{prop:comp}.  
Let $N\in\mathbb{N}$ and let $P_N:\mathcal{H}_a\to \mathcal{H}_a$
denote the orthogonal projection onto the 
subspace spanned by $\{e_n : |n|  \leq N\}$. 
By Lemma \ref{lemc}, the spectrum of the  
finite rank operator $P_N\mathcal{C}P_N$ is given by 
\begin{equation}\label{specfinsec}
\sigma(P_N\mathcal{C}P_N) = \{ (-\lambda)^k: k\in\{1,\ldots,N\} \} \cup 
\{ (-\bar{\lambda})^k: k\in\{1,\ldots, N\} \}
\cup \{1,0\}\,.
\end{equation}
Moreover, each non-zero element of the spectrum 
of $P_N\mathcal{C}P_N$ 
is an eigenvalue the algebraic multiplicity of which 
coincides with the number of times the non-zero number 
occurs in (\ref{specfinsec}). 
Now, in order to finish the proof we only need to show that the nonzero
spectrum (with algebraic multiplicities) of the transfer 
operator $\mathcal{C}$ is captured by the non-zero spectra of the
finite rank operators $P_N\mathcal{C}P_N$. 
This follows from a standard spectral approximation result 
(see, for example, \cite[XI.9.5]{DS}) together with the fact 
that $P_N\mathcal{C}P_N$ converges to 
$\mathcal{C}$ in the operator norm on $\mathcal{H}_a$, 
which in turn follows from the fact that $\mathcal{C}$ is compact 
(see, for example, \cite[Theorem~4.1]{ALL}).
\end{proof}

\section{Invariant measure and correlation decay}\label{sec4}

Since the map (\ref{ab}) is area preserving it is clear that Haar
measure $\mu$ on the torus is invariant under $T$. 
This invariance can also be cast in terms of spectral properties of the transfer operator. 
In order to see this,  we note that 
the constant function $e_0$ is the eigenfunction of the transfer
operator corresponding to the 
eigenvalue $1$, since $\mathcal{C}e_0 = e_0\circ T=e_0$. 
Furthermore, for $f$ a Laurent polynomial
we define the functional
\begin{equation}\label{da}
\ell_*(f)=\int_{\mathbb{T}^2}f\,d\mu= \int_{\mathbb{T}^2} f(z) \frac{dz_1}{2 \pi i z_1}
\frac{dz_2}{2 \pi i z_2} = f_0 \, .
\end{equation}
Using the definition of the norm (\ref{bd}) we have
$|\ell_*(f)|=|f_0 |\leq \|f\|_a$. Thus the functional $\ell_*$ is bounded on
the dense subset of Laurent polynomials and thus extends uniquely 
to a functional $\ell_* : \mathcal{H}_a \rightarrow \mathbb{C}$ 
on the entire space ${\cal H}_a$,
which for simplicity we denote by $\ell_*$ again.

Using the fact that the map $T$ preserves Haar measure $\mu$ on $\mathbb{T}^2$, 
we have for any Laurent polynomial $f$ the relation
$\ell_*(\mathcal{C} f)=\ell_*(f\circ T)=\ell_*(f)$ and by continuity
this identity carries over to the entire space as well. Hence $\ell_*$
is the left-eigenfunctional of the transfer operator corresponding to the
leading eigenvalue $1$. 

All in all, we can now define a bounded projection 
$\mathcal{P}:\mathcal{H}_a\to \mathcal{H}_a$ by setting 
\begin{equation}
\label{Pdef}
\mathcal{P}f=\ell_*(f)e_0\,, 
\end{equation}
which, by what has been said above, satisfies 
\begin{equation}
\label{Pprop}
\mathcal{C}\mathcal{P}=\mathcal{P}\mathcal{C}=\mathcal{P}\,, 
\end{equation}
which means that $\mathcal{P}$ is the spectral projection
corresponding to the eigenvalue 1.

We now turn to the study of correlation functions with respect to $\mu$. 
Let $g:\mathbb{T}^2 \rightarrow \mathbb{C}$ 
be analytic in a neighbourhood of the unit torus so that 
$|g_n|\leq c \exp(-\gamma |n|)$ for some $\gamma>0$, $c>0$. Choose $a$
sufficiently small so that $\exp(-\gamma |n|) \exp(a|n_u|-a|n_s|)$
decays exponentially in $|n|$. Define the functional
\begin{equation}\label{db}
\ell_g(f) = \int_{\mathbb{T}^2}fg\,d\mu=\int_{\mathbb{T}^2} g(z) f(z) \frac{dz_1}{2 \pi i z_1}
\frac{dz_2}{2 \pi i z_2}
=\sum_{n\in \mathbb{Z}^2} g_{-n} f_n
\end{equation}
on the dense subset of Laurent polynomials. Using the Cauchy-Schwarz inequality and the 
definition of the norm (\ref{bd}) we conclude that
\begin{equation}\label{dc}
|\ell_g(f)|^2 \leq \|f\|_a^2
\sum_{n\in\mathbb{Z}^2} |g_{-n}|^2 \exp(2a|n_u|-2a|n_s|) \, .
\end{equation}
Hence $\ell_g$ extends to a bounded functional on the
entire space ${\cal H}_a$. 

Thus, for any observable $h\in \mathcal{H}_a$, 
the correlation function (\ref{af}) 
can, using (\ref{da}), (\ref{Pdef}) and (\ref{db}), 
be cast into spectral form as follows:
\begin{equation}\label{dd}
C_{gh}(k)= \ell_g(\mathcal{C}^k h) - \ell_g(e_0) \ell_*(h) =
\ell_g(\mathcal{C}^kh-\mathcal{P}h)\,.
\end{equation}
Since the spectrum of $\mathcal{C}$ is discrete we have a partial
spectral 
decomposition (see, for example, 
\cite[Chapter~V, Theorem~9.2]{taylorlay}) of the form 
\begin{equation}\label{specdec}
\mathcal{C}^kh=\mathcal{P}h
+ \mathcal{R}^kh\,, 
\end{equation}
where 
$\mathcal{R}:\mathcal{H}_a\to \mathcal{H}_a$ 
is a compact operator 
with 
\begin{equation}
\label{Rspec}
\sigma(\mathcal{R})=\sigma(\mathcal{C})\setminus \{1\}\,,
\end{equation} 
which implies that 
the spectral radius of $\mathcal{R}$ is equal to $|\lambda|$, that is, 
\begin{equation}
\label{Rspecrad}
\lim_{k\to \infty}\|\mathcal{R}^k\|^{1/k}=|\lambda|\,.
\end{equation} 
Combining (\ref{dd}), (\ref{specdec}) and (\ref{Rspecrad}) 
we obtain the desired bound 
\begin{equation}
\limsup_{k\to \infty}|C_{gh}(k)|^{1/k}
=\limsup_{k\to \infty}|\ell_g(\mathcal{R}^kh )|^{1/k} 
\leq |\lambda|\,,
\end{equation}
since $\ell_g$ is bounded on $\mathcal{H}_a$. 
This furnishes the proof of Corollary~\ref{cora}.
It is quite easy to see that including lower-lying eigenvalues into the 
spectral decomposition (\ref{specdec}) we can obtain asymptotic 
expansions for the correlation function.

For the proof of Corollary~\ref{corb} we first note that (\ref{Pprop})
implies that for every natural number $k$ we have 
\begin{equation}
\label{cp}
\mathcal{C}^k-\mathcal{P}=\mathcal{R}^k\,.
\end{equation}
Using a Neumann series for the resolvent of
the compact operator $\mathcal{R}$ together with 
(\ref{cp}) and (\ref{dd}) we now obtain for all
$\zeta \in \mathbb{C}$ with $|\zeta|>1$ 
\begin{align}
\hat{C}_{gh}(\zeta) & = \sum_{k=0}^\infty \zeta^{-k} C_{gh}(k) \\
& = \sum_{k=0}^\infty \zeta^{-k} \ell_{g}(\mathcal{C}^kh-\mathcal{P}h)\\
& = \ell_g(\sum_{k=0}^\infty \zeta^{-k} \mathcal{R}^kh)\\
& = \zeta \ell_{g}((\zeta I-\mathcal{R})^{-1}h)\,.
\end{align}
The corollary now follows from (\ref{Rspec}) together with the
observation that the resolvent of the compact operator $\mathcal{R}$ 
is analytic on
the punctured plane $\mathbb{C}\setminus \{0\}$ except for poles at the
non-zero eigenvalues (see, for example, 
\cite[Chapter~V, Corollary~10.3]{taylorlay}).

\section{Conclusion}

Having access to explicitly solvable examples helps to understand
dynamical features and to test conjectures.
Our example demonstrates that in the analytic category hyperbolic
diffeomorphisms exist for which the corresponding transfer operator 
has infinitely many distinct eigenvalues. In addition, eigenvalues 
can be arbitrarily close to one in modulus.

The Hamiltonian structure, that is, 
the fact that we have considered an
area preserving diffeomorphism has simplified our arguments at a
technical level. In addition, the model considered here does not show
the generic decay of eigenvalues expected for two dimensional 
maps (see \cite{Nau}). 
It is, however, rather straightforward to analyse
more general models along the lines presented here to restore
the generic behaviour and to investigate cases with a non-trivial 
invariant measure.

Our setup has been tailored for the model under consideration.
We have chosen a special Hilbert space with equal weightings
and components according to the eigendirections of the cat map,
see (\ref{bb}). While these choices turned out to be successful
their precise meaning remained somehow obscure. In addition, we were
able to transform the matrix representation of the transfer operator to
a triangular structure which gave us access to the entire spectrum. 
All these features are not entirely coincidental, in the sense that there is an underlying
functional analytic structure. Uncovering this structure
requires a more general approach based
on more subtle functional analytic techniques. The focus of the present
contribution has been on an elementary rigorous study of a particular
example which should be accessible for a larger, non-specialised audience. The 
general theory for analytic diffeomorphism of the torus alluded to above 
will be presented elsewhere.


\appendix
\setcounter{section}{0}
\renewcommand\thesection{\Alph{section}}

\section{Some numerical findings}\label{appa}

A visual impression of the hyperbolic structure can be obtained by the numerical
computation of the unstable and stable manifold of the fixed point. Straightforward
forward and backward iteration gives a fairly robust algorithm for the
computation of a finite part of these manifolds, see Figure~\ref{figa}. Even 
though the invariant density is uniform the geometry of the hyperbolic structure
is apparently non-uniform, but this non-uniformity is compensated for by a
respective variation of the local expansion and contraction rates.

\begin{figure}[h]
\centering \includegraphics[width=0.5 \textwidth]{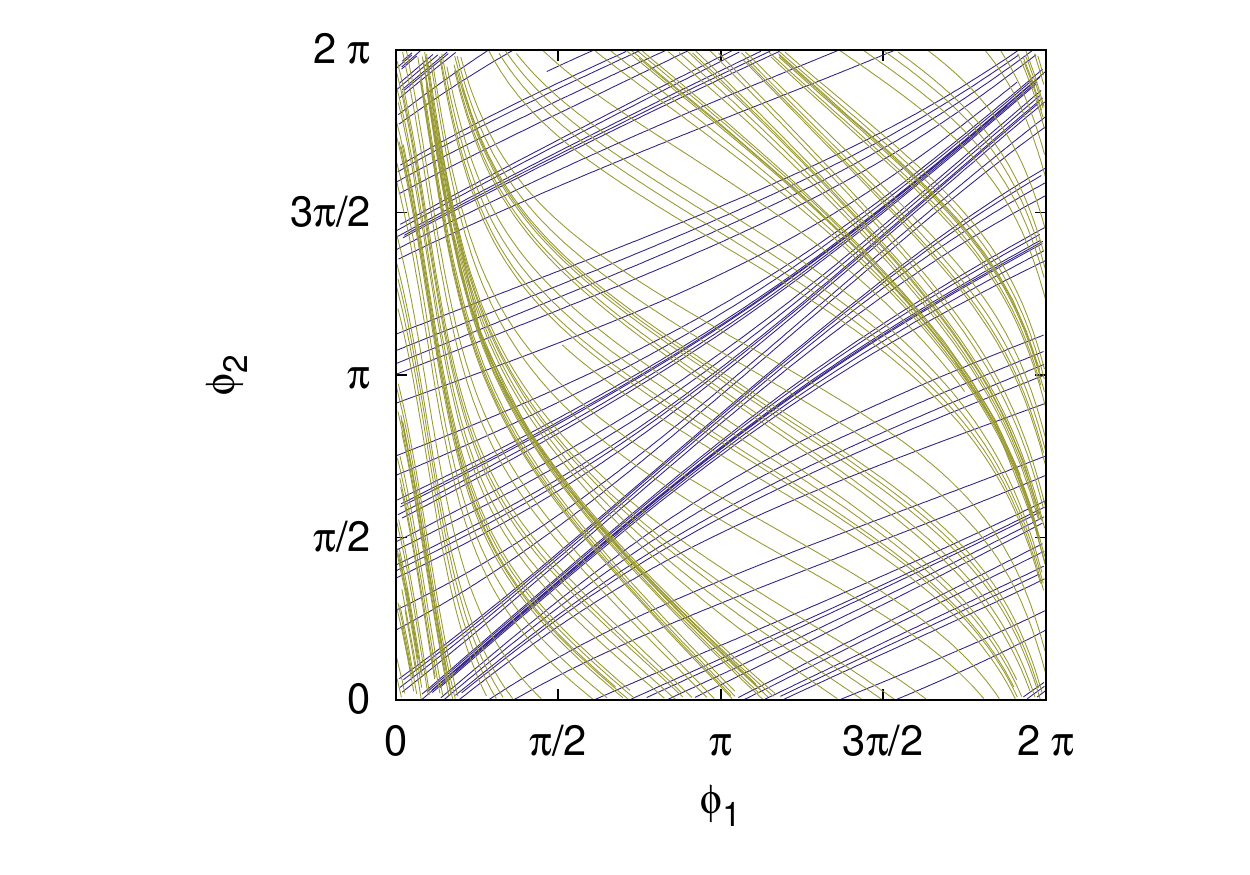}
\caption{Numerical result for the unstable (blue) and stable (bronze)
manifold of the map given by (\ref{ab}), for $\lambda=0.7 \exp(0.3i)$. 
\label{figa}}
\end{figure}

For simple trigonometric observables, the correlation function can
be computed directly. Consider, for instance, the case of the autocorrelation
of $\cos(\phi_2)$,  which corresponds to choosing 
$g(z_1,z_2)=h(z_1,z_2)=(z_2+z_2^{-1})/2$ in (\ref{af}). Since the 
invariant density is constant the mean values obviously vanish.
In order to compute the correlation integral we introduce the shorthand 
\begin{equation}\label{fa}
w_{n_1 n_2}(z_1,z_2)=z_1^{n_1} z_2^{n_2}
\end{equation}
for denoting monomials. By definition of the transfer operator, it follows 
that for a function $f$ which is analytic for $|z_\ell|<1$ the
expression $\mathcal{C}f=f\circ T$ is analytic for $|z_\ell|<1$
as well. An analogous property holds for functions which are analytic
for $|z_\ell|>1$. Hence, the correlation integral can be written as
\begin{equation}\label{fb}
C(k)=\frac{1}{4} \int_{\mathbb{T}^2}  \left(z_2^{-1}
(\mathcal{C}^k w_{01})(z_1,z_2) + z_2
(\mathcal{C}^k w_{0 -1})(z_1,z_2) \right) \frac{dz_1}{2\pi i z_1}
\frac{dz_2}{2\pi i z_2} \, .
\end{equation}
As for the action of the transfer operator on monomials, we see that 
\begin{align}
(\mathcal{C} w_{0 n_2})(z_1,z_2) = & (-\lambda)^{n_2} z_2^{n_2} +
\mathcal{O}(z_2^{n_2} z_1), \qquad (n_2 \geq 0) \label{fca} \\
(\mathcal{C} w_{0 n_2})(z_1,z_2) = & (-\bar{\lambda})^{-n_2} z_2^{n_2} 
+\mathcal{O}(z_2^{n_2} z^{-1}_1), \qquad (n_2 \leq 0) \label{fcb} \\
(\mathcal{C} w_{n_1 n_2})(z_1,z_2) = &  {\cal O}(z_2^{n_1+n_2} z_1^{n_1}), 
\qquad (n_1, n_2 \geq 0 \mbox{ or } n_1, n_2 \leq 0) \label{fcc}
\end{align}
where the higher order terms, as mentioned above, are analytic either
for $|z_\ell|<1$ or $|z_\ell|>1$. Hence, only the leading term in 
(\ref{fca}) and (\ref{fcb}) contributes to the integrals in 
(\ref{fb}) and we arrive at
\begin{equation}\label{fd}
C(k)= \frac{(-\lambda)^k + (-\bar{\lambda})^k}{4} \, .
\end{equation}
As a by-product we obtain that, as expected, the upper bound given by 
Corollary~\ref{cora} is sharp.

\begin{figure}[h]
\centering \includegraphics[width=0.5 \textwidth]{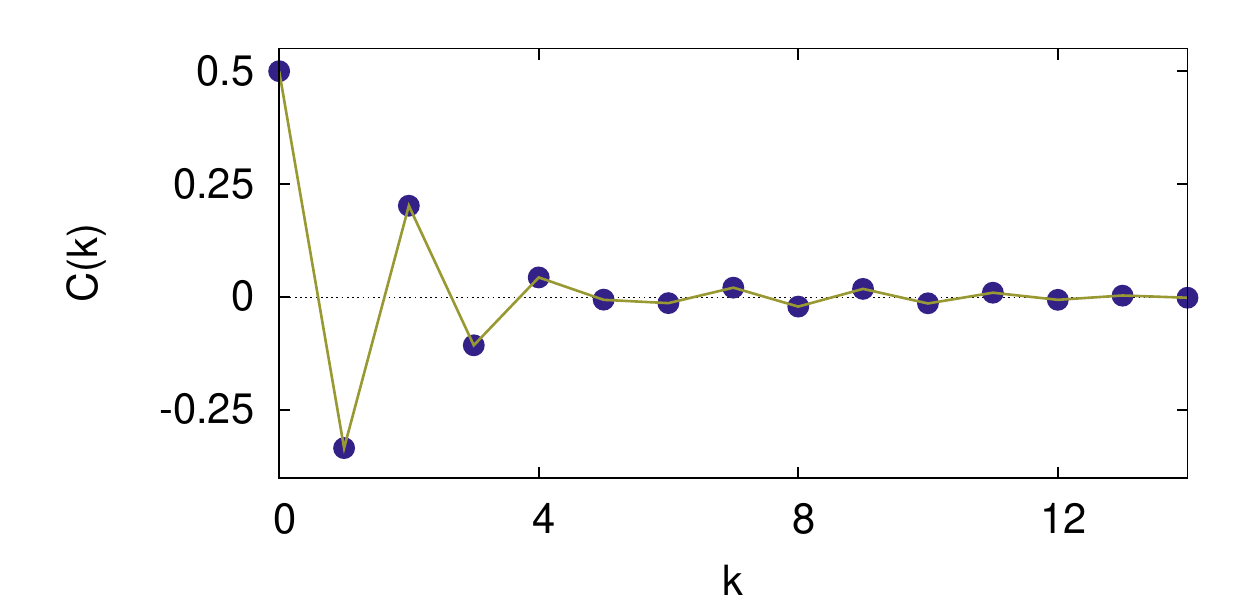}
\caption{Autocorrelation function of $\cos(\phi_2)$ according to
equation (\ref{fd}) (bronze line). 
Symbols (blue) are results of a numerical simulation
with an ensemble of $10^8$ data points.
\label{figb}}
\end{figure}

\section{A lower bound}
\label{appb}
This short appendix is devoted to proving a bound required in the proof of Lemma~\ref{lemb}. 
\begin{lem}\label{lemd}
Let $F\colon\mathbb{R}^2\to \mathbb{R}^2$ be given by
\begin{equation}
F(x,y)=2\left ( |\p x+1|-|\p^{-1}x-1| +|\p y-1|-
|\p^{-1}y+1| \right )  \, .
\end{equation}
Then for $x\geq y$, we have
\begin{equation}\label{ea}
 F(x,y) \geq 
  \begin{cases}
   x-y & \text{if $|y|< 2\p^{-1}$,}\\
   x-y + |y|/2 & \text{if $|y| \geq 2\p^{-1}$,}
\end{cases}
\end{equation}
where, as before, $\p=(1+\sqrt{5})/2$ denotes the golden mean. 
\end{lem}
\begin{proof}
We can write
\begin{equation} 
F(x,y) = G(x)-H(y),
\end{equation}
where
\begin{align}
G(x) =&  \begin{cases}
         -2x-4 & \text{ if $x\in (-\infty, -\p^{-1})$,}  \\
 2\sqrt{5} x & \text{ if $x\in [-\p^{-1}, \p]$,}  \\
          2x+4 & \text{ if $x\in (\p, +\infty)$,} 
 \end{cases} \label{ec1} \\
 H(y) =&  \begin{cases}
         2y-4 & \text{ if $y\in (-\infty, -\p)$,}  \\
 2\sqrt{5} y & \text{ if $y\in [-\p, \p^{-1}]$,}  \\
          -2y+4 & \text{ if $y\in (\p^{-1}, +\infty)$.} 
 \end{cases} \label{ec2} 
\end{align}
Let  $x_m = -\p^{-1}$ denote the minimum of $G$
and let $y_m = - x_m$ denote the maximum of $H$.  
We start by showing that for $x\geq y$ we have 
\begin{equation} \label{ghgreater}
F(x,y) \geq 2(x-y)\,.
\end{equation}
In order to see this, we first observe that 
$G(t) - H(t) \geq 0$ for all $t\in \mathbb{R}$
and that the minimal slope in modulus of $G$ and $H$ is $2$. Moreover, we have 
$G(x)\geq 2x$ for $x\geq y_m$ and $H(y)\leq 2y$ for $y \leq x_m$. 

Now, for $x \geq x_m$ we have 
$G'(x) \geq 2 $, so that 
\begin{equation}
F(x,y) = (G(x)-G(y)) + (G(y)-H(y)) \geq 2 (x-y), 
\quad \text{for } x_m \leq  y \leq x\,,
\end{equation}
while for $y \leq x \leq  y_m$ we have  $H'(y) \geq 2 $ so 
\begin{equation}
F(x,y) = (G(x)-H(x)) + (H(x)-H(y)) \geq 2 (x-y)\,. 
\end{equation}
Finally, we note that if $x\geq y_m$ and $y\leq x_m$ then 
\begin{equation}
F(x,y) = G(x) - H(y) \geq 2x -2y =2(x-y)\,, 
\end{equation}
which proves (\ref{ghgreater}) and the first part of the lemma.

In order to prove the second part, observe that for $x, x_m \geq y$ we have  
\begin{equation}
G(x)-H(y) \geq \frac{1}{2} (G(x)-H(y)) + \frac{1}{2} (G(x_m)-H(y))
\geq (x-y) + (x_m - y). 
\end{equation}
Thus, for $2x_m \geq y$ we have $x_m - y \geq -y/2 = |y| / 2$, and hence
\begin{equation}
F(x,y) = G(x) - H(y) \geq x-y + \frac{|y|}{2}.
\end{equation}
Similarly, for $x \geq y \geq 2 y_m$ we have $G(x)-H(y) \geq (x-y)+(y-y_m)$
and $y-y_m > |y|/2$, hence again $F(x,y) \geq x-y + |y|/2$,
finishing the proof of the lemma.
\end{proof}

\section{A remark on general Blaschke products}
\label{appc}

For our specific example defined in (\ref{aa}), the asymptotic decay of eigenvalues 
does not follow the generic pattern expected for two-dimensional maps
(see \cite{Nau}). 
It is
nevertheless quite easy to come up with solvable models exhibiting this 
generic behaviour. If we recall that the cat map can be written as a 
composition of area preserving orientation reversing linear automorphisms, 
and if we deform this automorphism by introducing a Blaschke factor, that is, if we define 
\begin{equation}\label{gb}
S_\lambda(z_1,z_2)=\left(\frac{z_1-\lambda}{1-\bar{\lambda} z_1} z_2,
z_1 \right)\,,
\end{equation}
which is an area preserving diffeomorphism of the torus, then the composition
\begin{equation}\label{gc}
T=S_\lambda \circ S_\mu
\end{equation}
yields a two-parameter area preserving family. With the tools introduced 
previously it is possible, but extremely tedious, to show that the corresponding
transfer operator is compact on a suitably weighted Hilbert space.
Even the spectrum, can be evaluated in closed form consisting of
simple eigenvalues
$1$, $(-\lambda)^n$, $(-\bar{\lambda})^n$, $(-\mu)^n$, $(-\bar{\mu})^n$,
$(-\lambda)^n (-\mu)^m$,
$(-\lambda)^n (-\bar{\mu})^m$,
$(-\bar{\lambda})^n (-\mu)^m$, and
$(-\bar{\lambda})^n (-\bar{\mu})^m$ where $n\geq 1$ and $m \geq 1$.
A more conceptual proof of these assertions is possible, but requires 
fairly heavy machinery, to be presented elsewhere. Here we will simply
illustrate this result by numerical means.
For that purpose we compute a truncated matrix representation of the transfer 
operator by using the standard Fourier basis (see (\ref{bg})), and
apply a standard 
eigenvalue solver. The result is presented in Figure~\ref{figc}.

\begin{figure}[h]
\centering \includegraphics[width=0.5 \textwidth]{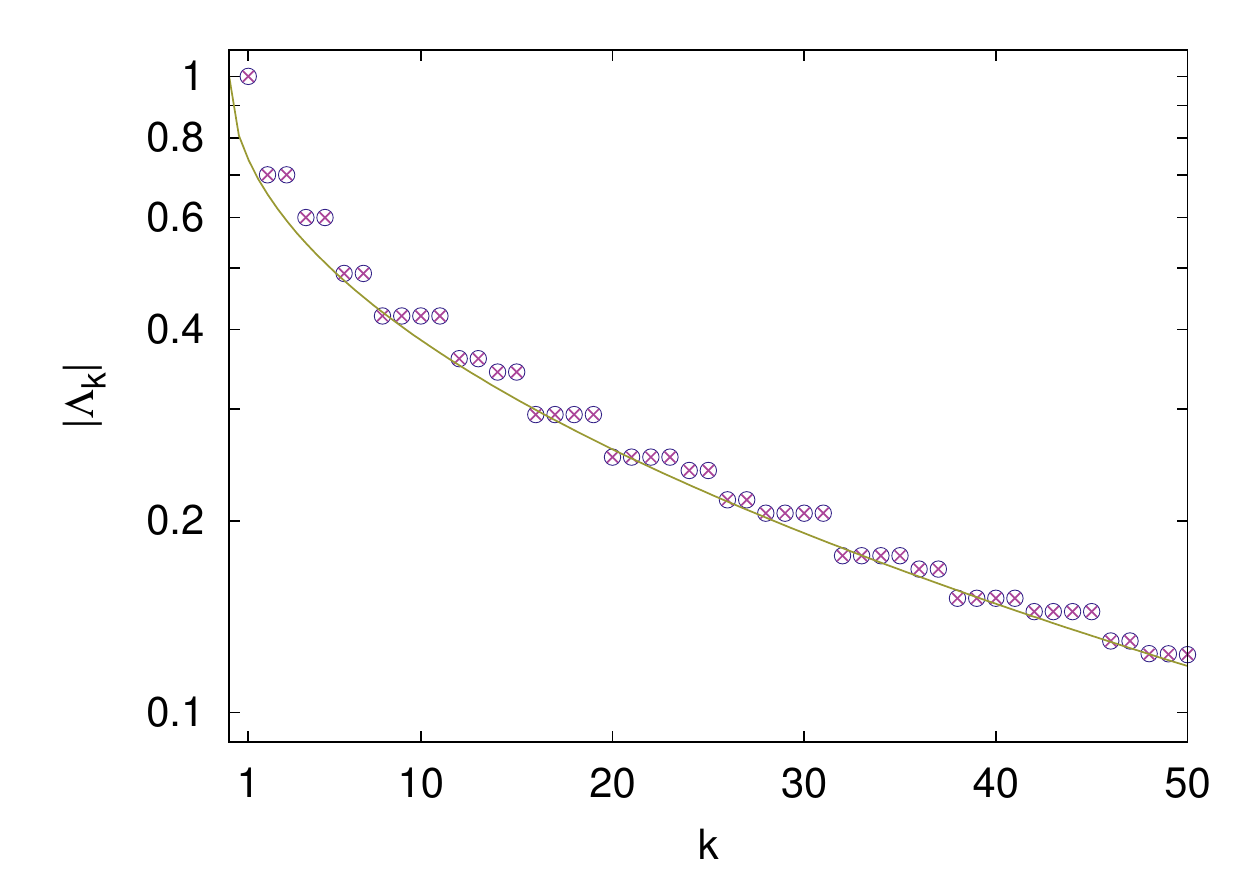}
\caption{Eigenvalues $\Lambda_k$ of the transfer operator corresponding to
the map (\ref{gc}) for $\lambda=0.7$ and $\mu=0.6$, ordered by size.
The numerical diagonalisation is illustrated by circles and the exact analytic
expression by crosses. The line indicates the generic 
asymptotic decay given by $|\Lambda_k| \sim \exp(-c \sqrt{k})$ with
$c=(\ln \lambda \ln \mu /2)^{1/2}$.
\label{figc}}
\end{figure}

For simplicity, we have so far considered 
area preserving maps where
the explicit expression for the invariant measure is known a priori. 
The invariant measures of two-dimensional Blaschke products exhibit a richer structure 
(see \cite{PuSh_ETDS08}).
The tools introduced here allow for a detailed study
of those measures. In a nutshell, maps where the determinant of the 
Jacobian vary may have fractal properties. Toy examples based on piecewise
linear maps are well established in the literature (see, for example, \cite{Neun_JPA02}).
Blaschke products offer a systematic and
analytic approach towards such features. 
For the purpose of illustration consider the simple model
\begin{equation}\label{gd}
T(z_1,z_2)=\left( z_1^2 \frac{z_2-\mu}{1-\bar{\mu} z_2}, 
z_1 \frac{z_2-\mu}{1-\bar{\mu} z_2} \right) \, .
\end{equation}
Our approach allows for a detailed investigation of
the spectral structures, but details turn out to be quite cumbersome.
Hence, to visualise the properties of the invariant measure we just
compute a histogram by a suitable numerical simulation, see
Figure~\ref{figd}. 

\begin{figure}[h]
\centering \includegraphics[width=0.5 \textwidth]{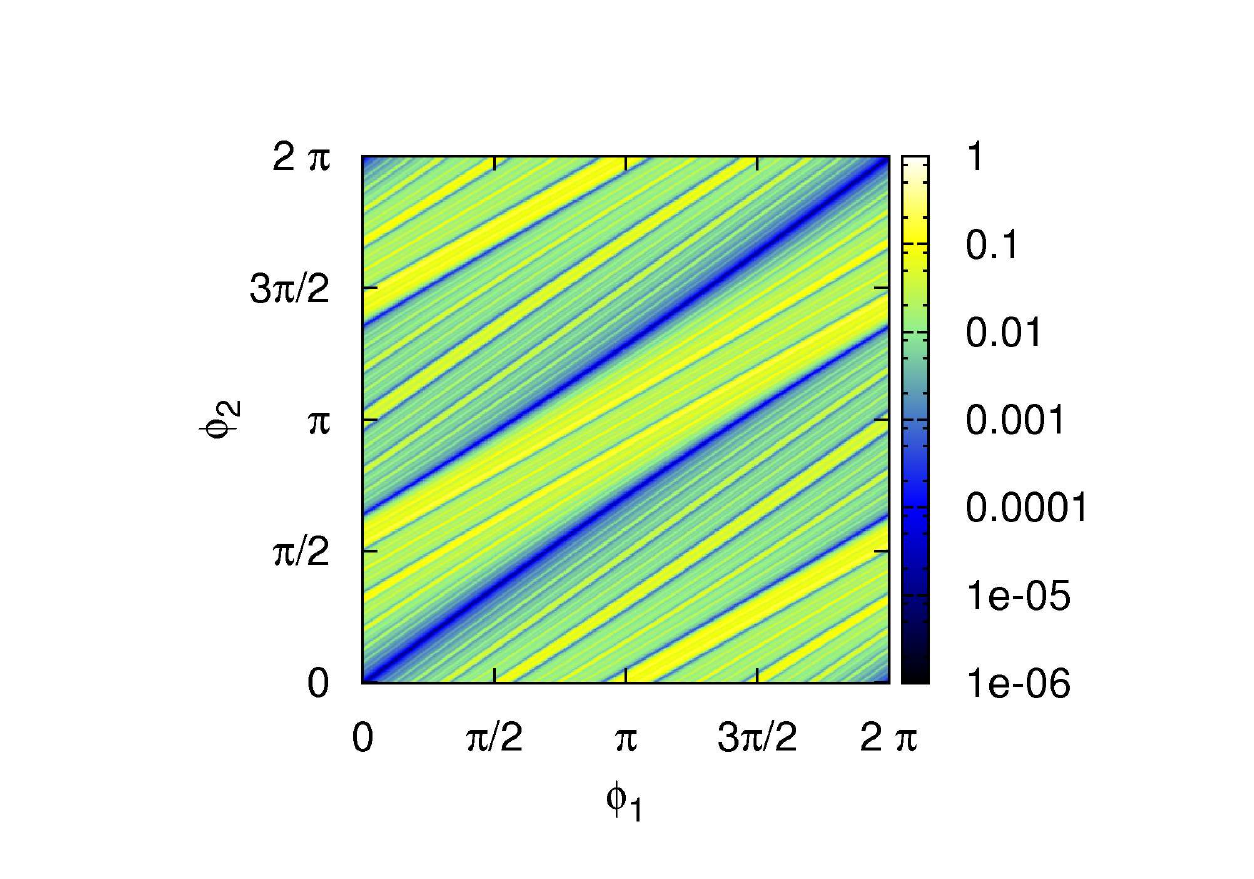}
\caption{Density plot illustrating the
invariant measure of the map given by (\ref{gd}) 
in real coordinates $z_\ell=\exp(i \phi_\ell)$ for $\mu=0.4$.
The data show a histogram with resolution $2\pi/5000 \times 2\pi/5000$
obtained from $10^4$ time traces of length
$2\times 10^7$ with uniformly distributed initial conditions.
\label{figd}}
\end{figure}


\end{document}